\documentclass[runningheads,a4paper]{llncs}

\usepackage{algorithm}
\usepackage{algpseudocode}
\usepackage{amsmath,amssymb}
\usepackage{footmisc}
\usepackage{graphicx}

\graphicspath{{Pictures/}} 

\usepackage{url}
\urldef{\mailsa}\path|{lijian83@mail, bw-zh14@mails, zhangny12@mails}.tsinghua.edu.cn|
\urldef{\mailsb}\path|haitao.wang@usu.edu|

\newcommand{\keywords}[1]{\par\addvspace\baselineskip
\noindent\keywordname\enspace\ignorespaces#1}

\newcommand{\calP}{\mathcal{P}}
\newcommand{\calL}{\mathcal{L}}
\newcommand{\MaxRS}{\mathsf{MaxCov_R}}
\newcommand{\MaxCRS}{\mathsf{MaxCov_D}}
\newcommand{\Cover}{\mathsf{Cover}}
\newcommand{\OPT}{\mathsf{Opt}}
\newcommand{\R}{\mathbb{R}}
\newcommand{\F}{\mathsf{F}}
\newcommand{\tF}{\widehat{\mathsf{F}}}

\newcommand{\cell}{\mathsf{c}}
\newcommand{\upperd}{\mathsf{b}}
\newcommand{\greedy}{\mathsf{Greedy}}

\newcommand{\calA}{\mathcal{A}}

\newcommand{\eat}[1]{}

\newcommand{\BlockPartion}{\textsc{MaxCovCell}}
\newcommand{\GeneralmPartion}{\textsc{MaxCovCellM}}

\newcommand{\DP}{\textsc{DP}}
\newcommand{\GREEDYP}{\textsc{Greedy}}

\newcommand{\Partition}{\textsc{Partition}}

\newcommand{\topic}[1]{\noindent{\underline{#1}:}}

\begin{document}

\mainmatter  

\title{Linear Time Approximation Schemes for Geometric Maximum Coverage\thanks{Jian Li, Bowei Zhang and Ningye Zhang's research was supported in part by the National Basic Research Program of China Grant 2015CB358700, 2011CBA00300, 2011CBA00301, the National Natural Science Foundation of China Grant 61202009, 61033001, 61361136003. Haitao Wang's research was supported in part by NSF under Grant CCF-1317143.}}


%
%
\author{Jian Li$^\dagger$
\and
Haitao Wang$^\ddagger$
\and
Bowei Zhang$^\dagger$
\and
Ningye Zhang$^\dagger$
}
%

\institute{$^\dagger$ Institute for Interdisciplinary Information Sciences (IIIS),  \\Tsinghua University, Beijing, China, 100084\\
$^\ddagger$ Department of Computer Science,	Utah State University, Utah, USA, 84322	
\\
\mailsa\\
\mailsb\\
}
%
%

\maketitle

\vspace{-0.3cm}

\begin{abstract}
We study approximation algorithms for the following geometric version of the maximum coverage problem:
Let $\calP$ be a set of $n$ weighted points in the plane.
We want to place $m$ $a \times b$ rectangles such that the sum of the weights of the points in $\calP$ covered by these rectangles is maximized.
For any fixed $\varepsilon>0$, we present efficient approximation schemes
that can find a $(1-\varepsilon)$-approximation to the optimal solution.
In particular, for $m=1$, our algorithm runs in linear time
$O(n\log (\frac{1}{\varepsilon}))$, improving over the previous result.
For $m>1$, we present an algorithm that runs in
$O(\frac{n}{\varepsilon}\log (\frac{1}{\varepsilon})+m(\frac{1}{\varepsilon})^{O(\min(\sqrt{m},\frac{1}{\varepsilon}))})$ time.
\end{abstract}

\vspace{-0.3cm}

\section{Introduction} 

The maximum coverage problem is a classic problem in theoretical computer science and combinatorial optimization.
In this problem, we are given a universe $\calP$ of weighted elements, a family of subsets and a number $k$.
 The goal is to select at most $k$ of these subsets such that the sum of the weights of the covered elements in $\calP$ is maximized.
It is well-known that the most natural greedy algorithm achieves
an approximation factor of $1-1/e$, which is essentially optimal (unless P=NP) \cite{hochbaum1998analysis,nemhauser1978analysis,feige1998threshold}.
However, for several geometric versions of the maximum coverage problem,
better approximation ratios can be achieved (we will mention some of such results below).
In this paper, we mainly consider the following geometric maximum coverage problem:
\begin{definition} ($\MaxRS(\calP,m)$)
Let $\calP$ be a set of $n$ points in a 2-dimensional Euclidean plane $\R^2$.
Each point  $p\in\calP$  has a given weight $w_{p} \geq 0$.
The goal of our geometric max-coverage problem (denoted as $\MaxRS(\calP,m)$) is to place $m$ $a \times b$ rectangles
such that the sum of the weights of the covered points by these rectangles is maximized.
More precisely, let $S$ be the union of $m$ rectangles we placed. Our goal is to maximize
$$
\Cover(\calP,S)=\sum_{p \in \calP\cap S}{w_{p}}.
$$
\end{definition}

We also study the same coverage problem with unit disks, instead of rectangles.
We denote the corresponding problem as $\MaxCRS(\calP, m)$.
One natural application of the geometric maximum coverage problem is the facility placement problem.
In this problem, we would like to locate a certain number of facilities to serve the maximum number of clients.
Each facility can serve a region (depending on whether the metric is $L_1$ or $L_2$,
the region is either a square or a disk).

\vspace{-0.3cm}
\subsection{$m=1$}
\vspace{-0.2cm}

\topic{Previous Results}
We first consider $\MaxRS(\calP,1)$.
Imai and Asano~\cite{imai1983finding}, Nandy and Bhattacharya~\cite{nandy1995unified} gave two different exact algorithms for  $\MaxRS(\calP,1)$, both running in time $O(n\log n)$.
It is also known that solving $\MaxRS(\calP,1)$ exactly in algebraic decision tree model requires $\Omega(n\log n)$ time \cite{ben1983lower}.
Tao et al.~\cite{tao2013approximate} proposed a randomized approximation scheme for $\MaxRS(\calP,1)$.
With probability $1-1/n$, their algorithm returns a ($1-\varepsilon$)-approximate answer in
$O(n\log (\frac{1}{\varepsilon})+n\log\log n)$ time.
In the same paper, they also studied the problem in the external memory model.

\vspace{0.1cm}
\topic{Our Results}
For $\MaxRS(\calP,1)$ we show that there is an approximation scheme
that produces a ($1-\varepsilon$)-approximation and runs in $O(n\log (\frac{1}{\varepsilon}))$ time,
improving the result by Tao et al.~\cite{tao2013approximate}.

\vspace{-0.4cm}
\subsection{General $m>1$}

\topic{Previous Results}
Both $\MaxRS(\calP, m)$ and $\MaxCRS(\calP, m)$ are NP-hard if $m$ is part of the input \cite{megiddo1984complexity}.
The most related work is de Berg, Cabello and Har-Peled~\cite{de2009covering}.
They mainly focused on using unit disks (i.e., $\MaxCRS(\calP, m)$).
They proposed a $(1-\varepsilon)$-approximation algorithm for $\MaxCRS(\calP,m)$
with time complexity
$O(n(m/\varepsilon)^{O(\sqrt{m})})$.

\footnote{
They were mainly interested in the case where $m$ is a constant.
So the running time becomes
$O(n(1/\varepsilon)^{O(\sqrt{m})})$
(which is the bound claimed in their paper)
and
the exponential dependency on $m$ does not look too bad for $m=O(1)$.
Since we consider the more general case,
we make the dependency on $m$ explicit.
}
We note that their algorithm can be easily extended to $\MaxRS$ with the same time complexity.

We are not aware of any explicit result for $\MaxRS(\calP,m)$ for general $m>1$.

It is known \cite{de2009covering} that the problem admits a PTAS  via the standard shifting technique \cite{hochbaum1985approximation}.
\footnote{
Hochbaum and Maass \cite{hochbaum1985approximation}
obtained a PTAS for the problem of covering given points
with a minimal number of rectangles. Their algorithm can be easily modified into a PTAS
for $\MaxRS(\calP,m)$ with running time $n^{O(1/\epsilon)}$.
}

\vspace{0.1cm}
\topic{Our Results}
Our main result is an approximation scheme for $\MaxRS(\calP,m)$ which runs in time
$$
O\left( \frac{n}{\varepsilon}\log \frac{1}{\varepsilon}+m\left(\frac{1}{\varepsilon}\right)^{\Delta}\right),
$$
where $\Delta=O(\min(\sqrt{m},\frac{1}{\varepsilon}))$.
Our algorithm can be easily extended to other shapes. In Appendix~\ref{sec:othershape}, we sketch an extension for approximating
$\MaxCRS(\calP, m)$.
The running time of our algorithm is
$$
O\left(n\Bigl(\frac{1}{\varepsilon}\Bigr)^{O(1)}+m\Bigl(\frac{1}{\varepsilon}\Bigr)^{\Delta}\right).
$$

Following the convention of approximation algorithms,
$\varepsilon$ is a fixed constant.
Hence, the second term is essentially $O(m)$ and
the overall running time is essentially linear $O(n)$. Our algorithm follows the standard shifting technique~\cite{hochbaum1985approximation},
which reduces the problem to a smaller problem restricted in a constant size cell.
The same technique is also used in de Berg et al.~\cite{de2009covering}.
They proceeded by first solving the problem exactly in each cell, and then
use dynamic programming to find the optimal
allocation for all cells.
\footnote{
	In fact, their dynamic programming runs in time at least $\Omega(m^2)$.
	Since they focused on constant $m$, this term is negligible in their running time.
	But if $m>\sqrt{n}$, the term can not be ignored and may become the dominating term.
	}

Our improvement comes from another two simple yet useful ideas.
First, we apply the shifting technique in a different way and make the side length of grids much smaller ($O(\frac{1}{\varepsilon})$, instead of $O(m)$ in de Berg et al.'s algorithm~\cite{de2009covering}). Second, we solve the dynamic program approximately. In fact, we show that
a simple greedy strategy (along with some additional observations) can be used for this purpose, which allows us to save another $O(m)$ term.

\subsection{Other Related Work}

There are many different variants for this problem.
We mention some most related problems here.

Barequet et al. \cite{barequet1997translating}, Dickerson and Scharstein \cite{dickerson1998optimal} studied the max-enclosing polygon problem which aims to find a position of a given polygon to cover maximum number of points.
This is the same as $\MaxRS(\calP,1)$ if a polygon is a rectangle. Imai et al.~\cite{imai1983finding} gave an optimal algorithm for the max-enclosing rectangle problem with time complexity $O(n\log n)$.

$\MaxCRS(\calP, m)$ was introduced by Drezner~\cite{drezner1981note}.
Chazelle and Lee~\cite{chazelle1986circle} gave an $O(n^{2})$-time exact algorithm for the problem $\MaxCRS(\calP,1)$.
A Monte-Carlo $(1-\varepsilon)$-approximation algorithm for $\MaxCRS(\calP,1)$
was shown in \cite{agarwal2002translating}, where $\calP$ is
an unweighted point set.
Aronov and Har-Peled~\cite{aronov2008approximating} showed that for unweighted point sets an $O(n\varepsilon^{-2}\log n)$ time Monte-Carlo $(1-\varepsilon)$-approximation algorithm exists,
and also provided some results for other shapes. de Berg et al.~\cite{de2009covering} provided an $O(n\varepsilon^{-3})$ time $(1-\varepsilon)$-approximation algorithm.

For $m>1$, $\MaxCRS(\calP,m)$ has only a few results.
For $m=2$, Cabello et al.~\cite{cabello2008covering} gave an exact algorithm for
this problem when the two disks are disjoint in $O(n^{8/3}\log^{2}n)$ time. de Berg et al.~\cite{de2009covering} gave $(1-\varepsilon)$-approximation algorithms
that run in $O(n\varepsilon^{-4m+4}\log^{2m-1}{(1/\varepsilon)})$ time for $m>3$ and in $O(n\varepsilon^{-6m+6}\log{(1/\varepsilon)})$ time for $m=2,3$.

The dual of the maximum coverage problem is the classical set
cover problem. The geometric set cover problem
has enjoyed extensive study in the past two decades.
The literature is too vast to list exhaustively here.
See e.g., \cite{Bronnimann,Clarkson,even2005hitting,mustafa2009ptas,varadarajan2010weighted,Chan2012} and the references therein.

\vspace{-0.3cm}

\section{Preliminaries} 
\label{sec:prel}

We first define some notations and mention some results that are needed in our algorithm.
Denote by $G_{\delta}(a,b)$ the square grid with mesh size $\delta$
such that the vertical and horizontal lines  are defined as follows
\begin{equation*}
G_{\delta}(a,b)=\left\{(x,y)\in\mathbb{R}^{2}\mid y=b+k\cdot\delta,k\in \mathbb{Z}\right\}
\cup\left\{(x,y)\in\mathbb{R}^{2}\mid x=a+k\cdot\delta,k\in \mathbb{Z}\right\}
\end{equation*}
Given $G_{\delta}(a,b)$ and a point $p=(x,y)$, we call the integer pair
$(\lfloor x/\delta\rfloor,\lfloor y/\delta\rfloor)$ the {\em index} of $p$
(the index of the cell in which $p$ lies in).

\vspace{0.2cm}
\topic{Perfect Hashing}
Dietzfetbinger et al. \cite{dietzfelbinger1997reliable}
shows that if each basic algebraic operation (including $\{+,-,\times,\div,\log_2,\exp_2\}$) can be done in constant time,
we can get a perfect hash family so that each insertion and membership query takes $O(1)$ expected time.
In particular, using this hashing scheme,
we can hash the indices of all points, so that we can obtain the list of all non-empty cells in $O(n)$ expected time.
Moreover, for any non-empty cell, we can retrieve all points lies in it in time linear in the number of such points.

\vspace{0.2cm}
\topic{Linear Time Weighted Median and Selection} 
It is well known that finding the weighted median for an array of numbers can be done in deterministic
worst-case linear time.
The setting is as follows:
Given $n$ distinct elements $x_{1},x_{2},...,x_{n}$ with positive weights $w_{1},w_{2},...,w_{n}$.
Let $w=\sum_{i=1}^{n}w_{i}$.
The {\em weighted median} is the element $x_{k}$
satisfying $\sum_{x_{i}<x_{k}}w_{i}<w/2$ and $\sum_{x_{i}>x_{k}}w_{i}\leq w/2$.
Finding the kth smallest elements for any array can also be done in deterministic
worst-case linear time. See e.g., \cite{clrs}.

\vspace{0.2cm}
\topic{An Exact Algorithm for $\MaxRS(\calP, 1)$} As we mentioned previously, Nandy and Bhattacharya \cite{nandy1995unified} provided an $O(n\log n)$ exact algorithm for the $\MaxRS(\calP,1)$ problem. We are going to use this algorithm as a subroutine in our algorithm.

\vspace{-0.3cm}

\section{A Linear Time Algorithm for $\MaxRS(\calP,1)$}
\label{sec:m1}

\textbf{Notations:} Without loss of generality, we can assume that $a=b=1$, i.e., all the rectangles are $1\times1$ squares,
(by properly scaling the input).
We also assume that all points are in general positions.
In particular, all coordinates of all points are distinct.
For a unit square $r$, we use $w(r)$
to denote the sum of the weights of the points covered by $r$.
We say a unit square $r$ is located at $(x,y)$ if the top-left corner of $r$ is $(x,y)$.

Now we present our approximation algorithm for $\MaxRS(\calP,1)$.

\subsection{Grid Shifting} 
Recall the definition of a grid $G_{\delta}(a,b)$ (in Section~\ref{sec:prel}).
Consider the following four grids: $G_{2}(0,0)$,$G_{2}(0,1)$,$G_{2}(1,0)$,$G_{2}(1,1)$ with $\delta = 2$.
We can easily see that for any unit square $r$, there exists one of the above grids that does not intersect $r$
(i.e., $r$ is inside some cell of the grid).
This is also the case for the optimal solution.

Now, we describe the overall framework, which is similar to that in~\cite{tao2013approximate}.
Our algorithm differs in several details.
\BlockPartion($\cell$) is a subroutine that takes a $2\times2$ cell $\cell$ as input and returns a unit square $r$
that is a (1-$\varepsilon$)-approximate solution if the problem is restricted to cell $\cell$.
We present the details of \BlockPartion \ in the next subsection.

\vspace{0.3cm}
\begin{algorithm}[h]
  \caption{$\MaxRS(\calP, 1)$}
  \begin{algorithmic}[]
  \State $w_{\max}\leftarrow0$
  \For {each $G\in\{G_{2}(0,0),G_{2}(0,1),G_{2}(1,0),G_{2}(1,1)\}$}
  \State Use perfect hashing to find all the non-empty cells of $G$.
     \For {each non-empty cell $\cell$ of $G$}
     \State $r\leftarrow$ \BlockPartion($\cell$).
     \State {\bf If} $w(r)> w_{\max}$, {\bf then} $w_{\max}\leftarrow w(r)$ and $r_{\max}\leftarrow r$.
     \EndFor;
  \EndFor;
  \State \textbf{return} $r_{\max}$;
  \end{algorithmic}
\label{algo:mainalgo1}
\end{algorithm}
\vspace{-0.5cm}

As we argued above, there exists a grid $G$ such that the optimal solution is inside some cell $\cell^\star \in G$.
Therefore, $\BlockPartion(\cell^\star)$ should return a (1-$\varepsilon$)-approximation for the original problem $\MaxRS(\calP,1)$.

\vspace{-0.3cm}

\subsection{\BlockPartion} 

\vspace{-0.2cm}

\label{subsec:partition}
In this section, we present the details of the subroutine \BlockPartion.
Now we are dealing with the problem restricted to a single $2\times2$ cell $\cell$.
Denote the number of point in $\cell$ by $n_{\cell}$, and the sum of the weights of points
in $\cell$ by $W_{\cell}$.
We distinguish two cases, depending on
whether $n_{\cell}$ is larger or smaller than $\left(\frac{1}{\varepsilon}\right)^{2}$.
If $n_{\cell}<\left(\frac{1}{\varepsilon}\right)^{2}$, we simply apply the $O(n\log n)$ time exact algorithm. \cite{nandy1995unified}

The other case requires more work.
In this case, we further partition cell $\cell$ into many smaller cells.
First, we need the following simple lemma.

\begin{lemma}
\label{lm:partition}
Given $n$ points in $\R^2$ with positive weights $w_{1},w_{2},...,w_{n},$ $\sum_{i=1}^{n}w_{i}=w$.
Assume that $x_{1},x_{2},...,x_{n}$ are their distinct $x$-coordinates.
We are also given a value $w_{d}$ such that $\max(w_{1},w_{2},...,w_{n})\leq w_{d} \leq w$,
Then, we can find at most $2w/w_{d}$ vertical lines such that the sum of the weights of points strictly between (we do not count the points on these lines) any two adjacent lines is at most $w_{d}$
in time $O(n\log(w/w_{d}))$.
\end{lemma}

\vspace{-0.5cm}

\begin{algorithm}[t]
  \caption{\Partition($\{x_{1},x_{2},...,x_{n}\}$)}
  \begin{algorithmic}[]
  \State Find the weighted median $x_{k}$ (w.r.t. $w$-weight);
  \State $\calL=\calL\cup\{x_{k}\}$;
  \State Generate $S=\{x_{i}\mid w_{i}<x_{k}\}$, $L=\{x_{i}\mid w_{i}>x_{k}\}$;
  \State If the sum of the weights of the points in $S$ is lager than $w_{d}$, run \Partition(S);
  \State If the sum of the weights of the points in $L$ is lager than $w_{d}$, run \Partition(L);
  \end{algorithmic}
  \label{algo:partition}
\end{algorithm}

\begin{proof}
See Algorithm~\ref{algo:partition}.
In this algorithm, we apply the weighted median algorithm recursively.
Initially we have a global variable $\calL=\emptyset$, which upon termination is
the set of $x$-coordinates of the selected vertical lines.
Each time we find the weighted median $x_{k}$ and separate the point with the vertical line $x=x_{k}$,
which we add into $\calL$.
The sum of the weights of points in either side is at most half of the sum of the weights of all the points.
Hence, the depth of the recursion is at most $\lceil\log(w/w_{d})\rceil$.
Thus, the size of $\calL$ is at most $2^{\lceil\log(w/w_{d})\rceil}\leq 2w/w_{d}$,
and the running time is $O(n\log(w/w_{d}))$.
\qed
\end{proof}

Now, we describe how to partition cell $\cell$ into smaller cells.
First we partition $\cell$ with some vertical lines.
Let $\calL_v$ to denote a set of vertical lines. Initially, $\mathcal{L}=\emptyset$.
Let $w_{d}=\frac{\varepsilon\cdot W_{\cell}}{16}$.
We find all the points whose weights are at least $w_d$.
For each such point, we add to $\calL_v$ the vertical line that passes through the point.
Then, we apply Algorithm~\ref{algo:partition}
to all the points with weights less than $w_d$.
Next, we add a set $\calL_h$ of horizontal lines in exactly the same way.

\begin{lemma}
The sum of the weights of points strictly between any two adjacent lines in $\calL_v$ is at most $w_d=\frac{\varepsilon\cdot W_{\cell}}{16}$.
The number of vertical lines in $\calL_v$ is at most $\frac{32}{\varepsilon}$.
Both statements hold for $\calL_h$ as well.
\end{lemma}
\begin{proof}
The first statement is straightforward from the description of the algorithm.
We only need to prove the upper bound of the number of the vertical lines.
Assume the sum of the weights of those points considered in the first (resp. second) step is $W_{1}$(resp. $W_{2}$), $W_{1}+W_{2}=W_{\cell}$.
The number of vertical lines in $\calL_v$ is at most
\begin{equation*}
W_{1}/\left(\frac{\varepsilon\cdot W_{\cell}}{16}\right)+
2W_{2}/\left(\frac{\varepsilon\cdot W_{\cell}}{16}\right)
\leq \frac{32}{\varepsilon}.
\end{equation*}
The first term is due to the fact that the weight of each point we found in the first step has weight at least $\frac{\varepsilon\cdot W_{\cell}}{16}$, and
the second term directly follows from Lemma~\ref{lm:partition}. \qed
\end{proof}

We add both vertical boundaries of cell $\cell$ into $\calL_v$
and both horizontal boundaries of cell $\cell$ into $\calL_h$.
Now $\calL=\calL_v\cup \calL_h$
forms a grid of size at most $(\frac{32}{\varepsilon}+2) \times (\frac{32}{\varepsilon}+2)$.
Assume
$\calL=\{(x,y)\in\mathbb{R}^{2} \mid y=y_{j},j\in \{1,...,u\}\}\cup\{(x,y)\in\mathbb{R}^{2}\mid x=x_{i},i\in \{1,...,v\}\}$,
with both $\{y_{i}\}$ and $\{x_{i}\}$ are sorted.
$\calL$ partitions $\cell$ into {\em small cells}.
The final step of our algorithm is simply enumerating all the unit squares located at $(x_{i},y_{j}),i\in \{1,...,u\},j\in \{1,...,v\}$,
and return the one with the maximum coverage.
However, computing the coverage exactly for all these unit squares is expensive.
Instead, we only calculate the weight of these unit square approximately as follows.
For each unit square $r$, we only count the weight of points that are in some small cell fully covered by $r$.
Now, we show this can be done in $O\left(n_{\cell}\log \left(\frac{1}{\varepsilon}\right)+\left(\frac{1}{\varepsilon}\right)^{2}\right)$ time.

After sorting $\{y_{i}\}$ and $\{x_{i}\}$, we can use binary search to identify which small cell each point lies in.
So we can calculate the sum of the weights of points at the interior, edges or corners of all small cells
in  $O(n_{\cell}\log \left(\frac{1}{\varepsilon}\right))$ times.

Thus searching the unit square with the maximum (approximate) coverage
can be done with a standard incremental algorithm in $O\left(\frac{1}{\varepsilon}\right)^{2}$ time. 

For completeness, we provide the details as follows.
In fact, the problem can be reduced to the following problem:
We have a $k\times k$ grid, each cell $(i,j)$ has a given weight $w(i,j)$.
We are also given two non-decreasing integer function
\footnote{[k]=\{1,2, ... ,k\}}
$X:[k]\rightarrow [k],Y: [k]\rightarrow [k]$
such that $X(i)\geq i, Y(i)\geq i$ for all $i\in [k]$.
The goal is to find $i^\star,j^\star$ such that
\begin{equation*}
I(i, j)=\sum_{i\leq i'\leq X(i)}\sum_{j\leq j'\leq Y(j)}w_{i',j'}
\end{equation*}
is maximized.
First,we define an auxiliary function $H(i,j)=\sum_{j\leq j'\leq Y(j)}w_{i,j'}$.
This can be computed separately for each column. We maintain two pointers $p$ and $q$, such that
$q=Y(p)$ holds through the process.
Initially, $p=1$ and $q=Y(1)$. In each iteration, we increment $p$ by 1 and $q$ by $Y(p+1)-Y(p)$.
Each $w(i,j)$ will be encountered at most twice.
So we can calculate all $H(i,j)$ values in $O\left(\frac{1}{\varepsilon}\right)^{2}$ time.
Then in the same way, We can use $H(i,j)$ values to calculate all $I(i,j)$ values in exactly the same way.
This takes $O\left(\frac{1}{\varepsilon}\right)^{2}$ time as well.
We return the maximum $I(i,j)$ value as the result for \BlockPartion($\cell$).

Putting everything together,
we conclude that if $n_{\cell}\geq\left(\frac{1}{\varepsilon}\right)^{2}$,
the running time of \BlockPartion($\cell$)\ is
$O\left(n_c\log \left(\frac{1}{\varepsilon}\right)+\left(\frac{1}{\varepsilon}\right)^{2}\right).$
We can conclude the main result of this section with the following theorem.

\begin{theorem}
	\label{thm:main1}
Algorithm~\ref{algo:mainalgo1} returns a (1-$\varepsilon$)-approximate answer
for $\MaxRS(\calP, 1)$
in $O(n\log \left(\frac{1}{\varepsilon}\right))$ time.
\end{theorem}

\begin{proof}
We only need to prove that \BlockPartion($\cell$)
returns a (1-$\varepsilon$)-approximation for cell $\cell$.
The case $n_{\cell}<\left(\frac{1}{\varepsilon}\right)^{2}$ is trivial since we apply the exact algorithm.
So we only need to prove the case of $n_{\cell}\geq\left(\frac{1}{\varepsilon}\right)^{2}$.

\begin{figure}[t]
\centering
\includegraphics[width=0.4\textwidth]{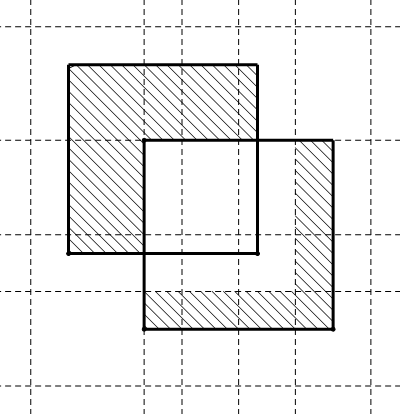}
\caption{Proof of Theorem~\ref{thm:main1}. Difference between the
	optimal solution(the top-left one) and our solution(the lower-right one)}
\label{fig:1}
\end{figure}

Suppose the optimal unit square is $r$. Denote by $\OPT$ the weight of the optimal solution.
The size of $\cell$ is $2\times 2$, so we can use $4$ unit squares to cover the entire cell.
Therefore, $\OPT\geq(\frac{W_{\cell}}{4})$.
Suppose $r$ is located at a point $p$, which is in the strict interior of a small cell $B$ separated by $\calL$.
\footnote{If $p$ lies on the boundary of $B$, the same argument still works.}
Suppose the index of $B$ is $(i,j)$.
We compare the weight of $r$ with $I(i,j)$ (which is the approximate weight of the unit square located at the top-left corner of $B$).
See Figure~\ref{fig:1}.
By the rule of our partition, the weight difference is at most $4$ times the maximum possible weight of points between two adjacent lines in $\calL$.
So $I(i,j)\geq \OPT-4\cdot \frac{\varepsilon\cdot W_{\cell}}{16}\geq (1-\varepsilon)\OPT$.
This proves the approximation guarantee of the algorithm.

Now, we analyze the running time.
 The running time consists of two parts:
 cells with number of points more than $\left(\frac{1}{\varepsilon}\right)^{2}$ and
 cells with number of points less than $\left(\frac{1}{\varepsilon}\right)^{2}$.
 Let $n_{1}\geq n_{2}\geq,...,\geq n_{j}\geq \left(\frac{1}{\varepsilon}\right)^{2}>n_{j+1}\geq n_{j+2},...,\geq n_{j+k}$
 be the sorted sequence of the number of points in all cells. Then, we have that
\begin{align*}
\text{Running time} \leq & \sum_{i=1}^{j}O\left(n_{i}\log \left(\frac{1}{\varepsilon}\right)+\left(\frac{1}{\varepsilon}\right)^{2}\right)+\sum_{i=1}^{k}O\left(n_{i+j}\log(n_{i+j})\right) \\
=
& O\left(\log \left(\frac{1}{\varepsilon}\right)\sum_{i=1}^{j}n_{i}+j\left(\frac{1}{\varepsilon}\right)^{2}+\sum_{i=1}^{k}n_{i+j}\log(n_{i+j})\right) \\
\leq & O\left(\log \left(\frac{1}{\varepsilon}\right)\sum_{i=1}^{j}(n_{i})+n+\sum_{i=1}^{k}n_{i+j}\log\left(\frac{1}{\varepsilon}\right)\right) \\
=    & O\left(\log \left(\frac{1}{\varepsilon}\right)\sum_{i=1}^{j+k}(n_{i})+n\right)
=O\left(n\log \left(\frac{1}{\varepsilon}\right)\right).
\end{align*}
This completes the proof. \qed
\end{proof}

\vspace{-0.3cm}

\section{Linear Time Algorithms for $\MaxRS(\calP,m)$} 
\label{sec:generalm}

\subsection{Grid Shifting}
For general $m$, we need the
shifting technique \cite{hochbaum1985approximation}.
Consider grids with a different side length:
$G_{6/\varepsilon}(a,b)$.
We shift the grid to $\frac{6}{\varepsilon}$ different positions: $(0,0),(1,1),....,(\frac{6}{\varepsilon}-1,\frac{6}{\varepsilon}-1)$.
(For simplicity, we assume that $\frac{1}{\varepsilon}$ is an integer and no point in $\calP$ has an integer coordinate, so points in $\calP$ will never lie on the grid line.
Let $$
\mathbb{G}=\left\{G_{6/\varepsilon}(0,0),...,G_{6/\varepsilon}(6/\varepsilon-1,6/\varepsilon-1)\right\}.
$$
Then we have the following lemma. 
\begin{lemma}
\label{shifting}
There exist $G^\star\in\mathbb{G}$ and a $(1-\frac{2\varepsilon}{3})$-approximate solution $R$
such that none of the unit squares in $R$ intersects $G^\star$.
\end{lemma}

\begin{figure}[t]
\centering
\includegraphics[width=0.6\textwidth]{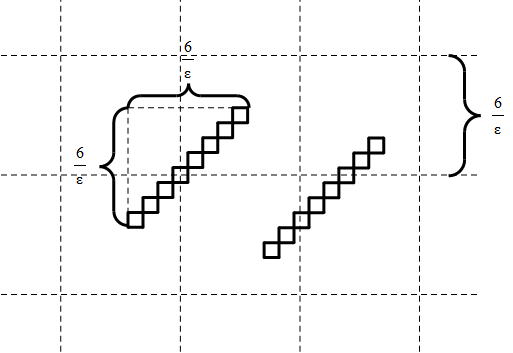}
\caption{Proof of Lemma~\ref{shifting}: the shifting technique.}
\label{trace}
\end{figure}

\begin{proof}
	For any point $p$, we can always use four unit squares to cover the $2\times2$ square centered at $p$.
	Therefore, there exists an optimal solution $OPT$ such that each covered point is cover by at most 4 unit squares in $OPT$.
	For each grid $G_{\frac{6}{\varepsilon}}(i,i)\in\mathbb{G}$, we build a modified answer $R_{i}$ from OPT in the following way. For each square $r$ that intersects with $G_{\frac{6}{\varepsilon}}(i,i)$, there are two different situations. If $r$ only intersects with one vertical line or one horizontal line. We move the square to one side of the line with bigger weight. In this case we will lose at most half of the weight of $r$. Notice that this kind of squares can only intersect with two grids in $\mathbb{G}$.
	Similarly, If $r$ intersects with one vertical line and one horizontal line at the same time, we move it to one of the four quadrants derived by these two lines. In this case we will lose at most 3/4 of the weight of $r$. This kind of squares can only intersect with one grid in $\mathbb{G}$. (see Figure~\ref{trace})
	Now we calculate the sum of the weights we lose from $R_{0},R_{2},...,R_{\frac{6}{\varepsilon}-1}$, which is at most $\max\{1/2\times2,3/4\times1\}=1$ times the sum of weights of squares in $OPT$. By the definition of $OPT$, it is at most $4w(OPT)$. So the sum of the weights of $R_{0},R_{2},...,R_{\frac{6}{\varepsilon}-1}$ is at least
	$(\frac{6}{\varepsilon}-4)w(OPT)$. Therefore there exists some $i$ such that $R_{i}$(which does not intersect $G_{\frac{6}{\varepsilon}}(i,i)$) is a $(1-\frac{2\varepsilon}{3})$ approximate answer.
	\qed
\end{proof}

We present a subroutine in section~\ref{subsec:F} which can approximately solve the problem for a grid,
and apply it to each non-empty grid in $\mathbb{G}$.
Then, in order to compute our final output from those obtained solutions, we apply a dynamic programming algorithm or a greedy algorithm which are shown in the next two sections.

\subsection{Dynamic Programming}
\label{subsec:dp}

Now consider a fixed grid $G\in \mathbb{G}$.
Let $\cell_{1}, \ldots , \cell_{t}$ be the cells of grid $G$
and $\OPT$ be the optimal solution that does not intersect $G$.
Obviously, $(\frac{6}{\varepsilon})^{2}$ unit squares are enough to cover an entire $\frac{6}{\varepsilon}\times\frac{6}{\varepsilon}$ cell.
Thus the maximum number of unit squares we need to place in one single cell is $m_{c}=\min\{m,(\frac{6}{\varepsilon})^{2}\}$.

Let $\OPT(\cell_i,k)$ be the maximum weight we can cover with $k$ unit squares in cell $\cell_i$.
For each nonempty cell $\cell_i$ and for each $k\in[m_{c}]$, we find a $(1-\frac{\varepsilon}{3})$-approximation $\F(\cell_i,k)$ to $\OPT(\cell_i,k)$. We will show how to achieve this later. Now assume that we can do it.

Let $\OPT_{\F}(m)$ be the optimal solution we can get from the values $\F(\cell_i,k)$.
More precisely,
\begin{align}
\label{eq:opt}
\OPT_{\F}(m)=\max_{k_1,\ldots,k_t \in [m_{c}]}
\left\{\sum_{i=1}^t \F(\cell_i,k_i)
	\,\,\Big\rvert\,\, \sum_{i=1}^{t} k_i=m\right\}
\end{align}
We can see that
$\OPT_{\F}(m)$ must be a $(1-\frac{\varepsilon}{3})$-approximation to $\OPT$.
We can easily use dynamic programming to calculate
the exact value of $\OPT_{\F}(m)$.
Denote by $A(i,k)$ the maximum weight we can cover with $k$ unit squares in cells $\cell_{1},\cell_{2},...,\cell_i$.
We have the following DP recursion:
\begin{equation*}
A(i,k)=\left\{
\begin{array}{lcl}
\max_{j=0}^{\min(k,m_{c})}
\left\{A(i-1,k-j)+\F(\cell_i,j)\right\}  &      &  {\text{if} \quad i>1}\\
\F(c_{1},k)  &      &  {\text{if} \quad i=1}\\
\end{array}\right.
\end{equation*}

The running time of the above simple dynamic programming is $O(m^{2}\cdot m_{c})$.
One may notice that each step of the DP is computing a $(+,\max)$ convolution.
However, existing algorithms (see e.g., \cite{bremner2006necklaces,williams2014faster}) only run slightly better than quadratic time. So the improvement would be quite marginal.
 But in the next section, we show that if we would like to settle for an
 approximation to $\OPT_{\F}(m)$,
 the running time can be dramatically improved to linear.

\subsection{A Greedy Algorithm}
\label{subsec:greedy}
We first apply our $\MaxRS(\calP, 1)$ algorithm in Section~\ref{sec:m1}
to each cell $\cell_i$,
to compute a $(1-\frac{\varepsilon^{2}}{9})$-approximation of
 $\OPT(\cell_i,1)$. Let $f(\cell_i,1)$ be the return values.
\footnote{
Both $f(\cell_i,1)$ and $\F(\cell_i,1)$ are approximations of $\OPT(\cell_i,1)$,
with slightly different approximation ratios.
}
This takes $O(n\log{\frac{1}{\varepsilon}})$ time.
Then, we use the selection algorithm to find out
the $m$ cells with the largest $f(\cell_i,1)$ values.
Assume that those cells are $\cell_{1}, ..., \cell_{m}, \cell_{m+1},...,\cell_{t}$,
sorted from largest to smallest by $f(\cell_i,1)$.

\begin{lemma}
\label{selectm}
Let $\OPT'$ be the maximum weight we can cover using $m$ unit squares in $\cell_{1}, ..., \cell_{m}$. Then
$\OPT'\geq(1-\frac{\varepsilon^{2}}{9})\OPT$
\end{lemma}
\begin{proof}
Let $k$ be the number of unit squares in $\OPT$ that are chosen from $\cell_{m+1},\ldots,\cell_{t}$. This means there must be at least $k$ cells in $\{\cell_{1},\ldots,\cell_{m}\}$ such that $\OPT$ does not place any unit square.
Therefore we can always move all $k$ unit squares placed in $\cell_{m+1},\ldots,\cell_{t}$ to these empty cells such that each empty cell contains only one unit square. Denote the weight of this modified solution by $A$. Obviously, $\OPT'\geq A$.
For any
$i$,$j$ such that $1\leq i\leq m<j\leq t$, we have $\OPT(\cell_i,1)\geq f(\cell_i,1)\geq f(\cell_j,1)\geq (1-\frac{\varepsilon^{2}}{9})\OPT(\cell_j,1)$. Combining with a simple observation that
$\OPT(\cell_i, k)\leq k\OPT(\cell_i, 1)$,
we can see that $A\geq(1-\frac{\varepsilon^{2}}{9})\OPT$. Therefore, $\OPT'\geq(1-\frac{\varepsilon^{2}}{9})\OPT$.
\qed
\end{proof}

Hence, from now on, we only need to consider the first $m$ cells
$\{\cell_{1},...,\cell_{m}\}$.
We distinguish two cases.
If $m\leq324(\frac{1}{\varepsilon})^{4}$,
we just apply the dynamic program to $\cell_{1},...,\cell_{m}$.
The running time of the above dynamic programming is $O((\frac{1}{\varepsilon})^{O(1)})$.

If $m>324(\frac{1}{\varepsilon})^{4}$,
we can use a greedy algorithm to find a answer of weight at least $(1-\frac{\varepsilon^{2}}{9})\OPT_{\F}(m)$.

Let $\upperd=(\frac{6}{\varepsilon})^2$.
For each cell $\cell_{i}$,
we find the upper convex hull of 2D points
$\{(0,\F(\cell_{i},0))$,$(1,\F(\cell_{i},1))$, \ldots , $(\upperd,\F(\cell_{i},\upperd))\}$.
See Figure~\ref{fig:fandtf}.
Suppose the convex hull points are
$\{(t_{i,0}, \F(\cell_{i},t_{i,0}))$, $(t_{i,1},\F(\cell_{i},t_{i,1}))$, ... , $(t_{i,s_{i}},\F(\cell_{i},t_{i,s_{i}}))\}$, where
$t_{i,0}=0$,$t_{i,s_{i}}=\upperd$.
For each cell, since the above points are already sorted from left to right, we can compute the convex hull in $O(\upperd)$ time by Graham's scan\cite{graham1972efficient}. Therefore, computing the convex hulls for all these cells takes $O(m \upperd)$ time.

For each cell $\cell_{i}$, we maintain a value $p_{i}$ representing that we are going to place $t_{i,p_{i}}$ squares in cell $\cell_{i}$.
Initially for all $i\in [m]$, $p_{i}=0$.
In each stage, we find the cell $\cell_{i}$ such that
current slope (the slope of the next convex hull edge)
$$
\frac{\F(\cell_{i},t_{i,p_{i}+1})-\F(\cell_{i},t_{i,p_{i}})}{t_{i,p_{i}+1}-t_{i,p_{i}}}$$
is maximized.
Then we add 1 to $p_{i}$, or equivalently
we assign $t_{i,p_{i}+1}-t_{i,p_{i}}$ more squares into cell $\cell_{i}$.
We repeat this step until we have already placed at least $m-\upperd$ squares.
We can always achieve this since we can place at most
$\upperd$ squares in one single cell in each iteration.
Let $m'$ the number of squares we have placed ($m=\upperd\leq m'\leq m$).
For the remaining $m-m'$ squares, we allocate them arbitrarily.
We denote the algorithm by $\GREEDYP$ and let the value obtained be $\greedy(m')$.
Having the convex hulls, the running time of the greedy algorithm is $O(m)$.

Now we analyze the performance of the greedy algorithm.

\begin{lemma}
\label{lm:greedyratio}
The above greedy algorithm computes an $(1-\varepsilon^2/9)$-approximation to
$\OPT_{\F}(m)$.
\end{lemma}

\begin{figure}[t]
	\centering
	\includegraphics[width=0.7\textwidth]{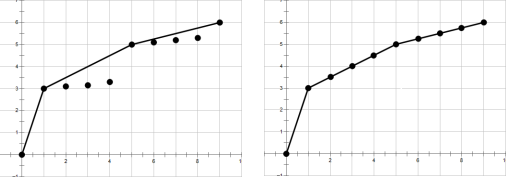}
	\caption{$\F(\cell_{i},k)$ (left) and $\tF(\cell_{i},k)$ (right)}
	\label{fig:fandtf}
\end{figure}

\begin{proof}
Define an auxiliary function $\tF(\cell_{i},k)$ as follows:
If $k=t_{i,j}$ for some $j$, $\tF(\cell_{i},k)=F(\cell_{i},k)$.
Otherwise, suppose $t_{i,j}<k<t_{i,j+1}$,
then $$
\tF(\cell_{i},k)=
F(\cell_{i},t_{i,j})+
\frac{F(\cell_{i},t_{i,j+1})-F(\cell_{i},t_{i,j})}{t_{i,j+1}-t_{i,j}}\times(k-t_{i,j}).
$$
Intuitively speaking, $\tF(\cell_{i},k)$(See Figure~\ref{fig:fandtf}) is the function defined by the
upper convex hull at integer points.
\footnote{
	At first sight, it may appear that $\F(\cell_i,k)$ should be a concave function.
	However, this is not true. See Figure~\ref{fig:convex} for an example.}
Thus, for all $i\in[m]$, $\tF(\cell_{i},k)$ is a concave function.
Obviously, $\tF(\cell_i,k)\geq \F(\cell_i,k)$ for all $i\in [m]$ and all $k\in [\upperd]$.

Let $\OPT_{\tF}(i)$  be the optimal solution we can get from the values $\tF(\cell_{i},k)$ by placing $i$ squares.
By the convexity of $\tF(\cell_{i},k)$,
the following greedy algorithm is optimal:
as long as we still have budget, we assign 1 more square to the cell
which provides the largest increment of the objective value.
In fact, this greedy algorithm runs in almost the same way
as \GREEDYP. The only difference is that  \GREEDYP\ only picks
an entire edge of the convex hull, while the greedy algorithm here
may stop in the middle of an edge (only happen for the last edge).
Since the marginal increment never increases, we
can see that $\OPT_{\tF}(i)$ is concave.

By the way of choosing cells in our greedy algorithm,
we make the following  simple but important observation:
$$
\greedy(m')=\OPT_{\tF}(m')=\OPT_{\F}(m').
$$
So, our greedy algorithm is in fact optimal for $m'$.
Combining with $m-m'\leq\upperd$ and the concavity of $\OPT_{\tF}$,
we can see that
$$
\OPT_{\tF}(m')\geq\frac{m-\upperd}{m} \OPT_{\tF}(m)\geq \left(1-\frac{\varepsilon^{2}}{9}\right)\OPT_{\F}(m).
$$
The last inequality holds because
$\OPT_{\tF}(i)\geq \OPT_{\F}(i)$ for any $i$.
\qed
\end{proof}

\begin{figure}[t]
	\centering
\includegraphics[width=0.8\textwidth]{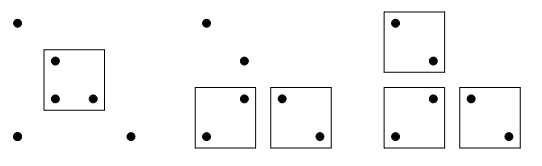}
	\caption{$\F(\cell_{i},k)$ may not be concave: $\F(\cell_{i},1)=3$, $\F(\cell_{i},2)=4$, $\F(\cell_{i},3)=6$.}
	\label{fig:convex}
\end{figure}

\subsection{Computing $\F(\cell,k)$ }
\label{subsec:F}
Now we show the subroutine $\GeneralmPartion$ for computing $\F(\cell,k)$.
We use a similar partition algorithm as Section~\ref{subsec:partition}.
The only difference is that this time we need to partition the cell finer
so that the maximum possible weight of points between any two adjacent parallel partition lines is $(\frac{\varepsilon^{3}W_{\cell}}{864})$.
After partitioning the cell, we enumerate all the possible ways of placing $k$
unit squares at the grid point.
Similarly, for each unit square $r$, we only count the weight of points that are in some cell fully covered by $r$. The algorithm is summarized in Algorithm~\ref{algo:mainalgo2}.

We can adapt the algorithm in~\cite{de2009covering} to enumerate these possible choices in
 $O((\frac{1}{\varepsilon})^{\Delta})$
time
where $\Delta=O(\min(\sqrt{m},\frac{1}{\varepsilon}))$.  We briefly sketch the algorithm in Appendix~\ref{appb}.
Now we prove the correctness of this algorithm.
\begin{lemma}
$\GeneralmPartion$ returns a $(1-\frac{\varepsilon}{3})$
approximate answer for $\OPT(\cell_i,k)$.
\end{lemma}
\begin{proof}
We can use $(\frac{6}{\varepsilon})^{2}$ unit squares to cover the entire cell, so $\OPT(\cell_i,k)\geq \frac{k\varepsilon^{2}W_{\cell}}{72}$.
By the same argument as in Theorem~\ref{thm:main1},
the difference between $\OPT(\cell_i,k)$
and the answer from Algorithm~\ref{algo:mainalgo2} are at most $4k$ times the maximum possible weight of points between two adjacent parallel partition lines.
Therefore, the algorithm returns a $(1-\frac{\varepsilon}{3})$-approximate answer of $\OPT(\cell_i,k)$.
\qed
\end{proof}
Now we can conclude the following theorem.
\begin{theorem}
\label{thm:main2}
Algorithm~\ref{algo:mainalgo2} returns a $(1-\varepsilon)$-approximate answer
	for $\MaxRS(\calP, m)$
	in time
	$$
	O\left( \frac{n}{\varepsilon}\log \frac{1}{\varepsilon}+m\left(\frac{1}{\varepsilon}\right)^{\Delta}\right),
	$$
	where $\Delta=O(\min(\sqrt{m},\frac{1}{\varepsilon}))$.
\end{theorem}

\begin{proof}

We know that $\OPT_{\F}(m)$ of $m$ selected cells is a $(1-\frac{\varepsilon}{3})$-approximation to $\OPT'$, so the greedy algorithm returns a $(1-\frac{\varepsilon}{3})(1-\frac{\varepsilon^{2}}{9})$-approximation to $\OPT'$. Combining with Lemma~\ref{shifting} and Lemma~\ref{selectm}, Algorithm~\ref{algo:mainalgo2} returns a $(1-\frac{2\varepsilon}{3})(1-\frac{\varepsilon}{3})(1-\frac{\varepsilon^{2}}{9})(1-\frac{\varepsilon^{2}}{9})$ approximate solution of the original problem.
Since $(1-\frac{2\varepsilon}{3})(1-\frac{\varepsilon}{3})(1-\frac{\varepsilon^{2}}{9})(1-\frac{\varepsilon^{2}}{9})>(1-\varepsilon)$, Algorithm~\ref{algo:mainalgo2} does return a $(1-\varepsilon)$-approximate solution.

\begin{algorithm}[h]
	\caption{$\MaxRS(\calP, m)$}
	\begin{algorithmic}[]
		\State $w_{\max}\leftarrow0$
		\For {each $G\in\{G_{\frac{6}{\varepsilon}}(0,0),...,G_{\frac{6}{\varepsilon}}(\frac{6}{\varepsilon}-1,\frac{6}{\varepsilon}-1)\}$}
		\State Use perfect hashing to find all the non-empty cells of $G$.
		\For {each non-empty cell $\cell$ of $G$}
		\State $r_{\cell}\leftarrow$ Algorithm~\ref{algo:mainalgo1} for $\cell$ with approximate ratio $(1-\frac{\varepsilon^{2}}{9})$
		\EndFor;
		\State Find the $m$ cells with the largest $r_{\cell}$.
		Suppose they are $\cell_{1},...,\cell_{m}$.
		\For{$i\leftarrow 1$ \textbf{to} $m$}
		\For{$k\leftarrow 1$ \textbf{to} $\upperd$}
		$\F(\cell_{i},k)\leftarrow$ \GeneralmPartion(c,k)
		\EndFor
		\EndFor;
		\State {\bf if} $m\leq324(\frac{1}{\varepsilon})^{4}$, {\bf then} $r\leftarrow \DP(\{\F(\cell_{i},k)\})$
		{\bf else}  $r\leftarrow \GREEDYP(\{\F(\cell_{i},k)\})$
		\State {\bf if} $w(r)> w_{\max}$, {\bf then} $w_{\max}\leftarrow w(r)$ and $r_{\max}\leftarrow r$
		\EndFor;
		\State \textbf{return} $r_{\max}$;
	\end{algorithmic}
    \label{algo:mainalgo2}
\end{algorithm}

We only need to calculate the overall running time. Solving the values $\F(\cell_i,1)$ and finding out the top $m$ results require $O(n\log \frac{1}{\varepsilon})$ time.
We compute the values $\F(\cell_i,k)$ of $m$ cells. For each cell $\cell_i$, we partition it only once and calculate $\F(\cell_i,1),\ldots,\F(\cell_i,\Delta)$ using the same partition.
Computing the values $\F(\cell_i,k)$ of all $m$ cells requires
$O(n\log (\frac{1}{\varepsilon})+m(\frac{1}{\varepsilon})^{\Delta})$
time.
We do the same for $\frac{1}{\varepsilon}$ different grids.
Therefore, the over all running time is as we state in the theorem.
\qed
\end{proof}

\vspace{-0.3cm}

\bibliographystyle{abbrv}
\bibliography{ref}

\appendix
\section{Enumeration In \GeneralmPartion}
\label{appb}
We can adapt the algorithm in~\cite{de2009covering} to enumerate these possible ways of placing $k$
unit squares at the grid point in
 $O((\frac{1}{\varepsilon})^{\Delta})$
time
where $\Delta=O(\min(\sqrt{m},\frac{1}{\varepsilon}))$.
We briefly sketch the algorithm.
We denote the optimal solution as $\OPT_\cell$.
From \cite{agarwal2002exact} we know that for any optimal solution,
there exists a line of integer grid (either horizontal or vertical) that
intersects with  $O(\sqrt{m})$ squares in $\OPT_\cell$, denoted as
the {\em parting line}.
So we can use dynamic programming.
At each stage, we enumerate the parting line, and the $O(\sqrt{m})$ squares
intersecting the parting line.
We also enumerate the number of squares in each side of the parting line in the optimal solution.
The total number of choices is $O((\frac{1}{\varepsilon})^{\Delta})$.
Then, we can solve recursively for each side. In the recursion,
we should consider that a subproblem which is composed of a smaller rectangle, and an enumeration of  $O(\sqrt{m})$ squares of the optimal solution intersecting the boundary of the rectangle and
at most $m$ squares fully contained in the rectangle.
Overall, the dynamic programming
can be carried out in $O((\frac{1}{\varepsilon})^{\Delta})$
time.

\section{Extension to Other Shapes}
\label{sec:othershape}
Our algorithm can be easily extended to other shapes, such as unit disks.
For ease of description, we only consider unit disks, i.e., the $\MaxCRS(\calP, m)$ problem.
The algorithm framework is almost the same as before,
except several implementation details required changing from rectangles to disks.
The major difference is the way for building an $\varepsilon$-approximation in each cell
(the partition scheme in Section~\ref{subsec:F} works only for rectangles).


Now, we sketch our algorithm, focusing only on the differences.
Again, we first adopt the shifting technique .
For unit disk, we consider grids with a different side length:
$G_{56/\varepsilon}(a,b)$.
Again for simplicity, we assume that $\frac{1}{\varepsilon}$ is an integer
and no point in $\calP$ has an integer coordinate.
We shift the grid to $\frac{28}{\varepsilon}$ different positions: $(0,0),(2,2),....,(\frac{56}{\varepsilon}-2,\frac{56}{\varepsilon}-2)$.
Let
$
\mathbb{G}=\left\{G_{56/\varepsilon}(0,0),...,G_{56/\varepsilon}(56/\varepsilon-2,56/\varepsilon-2)\right\}.
$
We can prove the following lemma, which is similar to Lemma~\ref{shifting}. The only difference in the proof
is that for unit disks, there exists an optimal answer $\OPT$
such that each covered point is covered by at most 7 disks in $\OPT$ instead of 4 for unit squares. 
\begin{lemma}
There exist $G^\star\in\mathbb{G}$ and a $(1-\frac{\varepsilon}{2})$-approximate answer $R$
such that all the disks in $R$
do not intersect any line in $G^\star$.
\end{lemma}

Now consider a fixed grid $G\in \mathbb{G}$.
For each cell $\cell$, we compute a $(1-\frac{\varepsilon^{2}}{8})$-approximation to $\OPT(\cell,1)$,
in $O(n\varepsilon^{-6})$ time,
by applying the following result in \cite{de2009covering}.
\begin{lemma} {\em \cite{de2009covering}}
\label{MaxCRSP1}
Give a set $\calP$ of $n$ weighted points and $\varepsilon(0<\varepsilon<1)$, we can find a $(1-\varepsilon)$-approximation for $\MaxCRS(\calP, 1)$  in $O(n\varepsilon^{-3})$ time.
\end{lemma}
Then, we use the selection algorithm to find out the $m$ cells with the largest returned values.
The rest of our algorithm will only consider those $m$ cells.
For each such cell $\cell$, we want to compute $(1-\frac{\varepsilon}{2})$-approximations $\F(\cell,k)$ of $\OPT(\cell,k)$.
For this purpose, we can use the following lemma from \cite{de2009covering}.

\begin{lemma}
\label{epsilonapprox2}{\em \cite{de2009covering}}
Give a set $\calP$ of n weighted points, an integer $m\geq1$ and $r\geq1$. Let $U$ be the collection of sets that are the union of m unit disks. In range space $(\calP,\{\calP\cap s\mid s\in U\})$,
Let $\calA$ be a $(1/2r)$-approximation for the range space,
where $r=w(P)/(\varepsilon\MaxCRS(\calP, m))$.
If $\OPT_{A}$ is a optimal solution for $\MaxCRS(A, m)$, then $\OPT_{A}$ is a
$(1-\varepsilon)$-approximation for $\MaxCRS(\calP, m)$.
\end{lemma}

By the above lemma, we only need to build a $(1/4r)$-approximation $\calA$ for the cell $\cell$,
where $r=W_{\cell}/(\varepsilon\OPT(\cell,k))$ and apply an exact algorithm to $\calA$.
Since we can use $O(1/\varepsilon^{2})$ unit disks to cover the whole cell,
we can see that $r<O(1/k\varepsilon^{3})$.
We can build the $1/4r$-approximation $A$ of size $O(\frac{k^{2}}{\varepsilon^{6}}\log\frac{k}{\varepsilon})$ in $O(n_{\cell}\cdot(\frac{k}{\varepsilon})^{O(1)})$ time (see e.g.,~\cite{fhktww-a-07}).
Using the exact algorithm in \cite{de2009covering} for $\MaxCRS(\calP, m)$,
we can compute the exact solution for $\calA$
in $O((\frac{k}{\varepsilon})^{O(\sqrt{k})})$ time.

Suppose we have computed all $\F(\cell,k)$ values for the first $m$ cells.
If $m\leq O(1/\varepsilon^{4})$, we apply the dynamic programming, as in Section~\ref{subsec:dp}.
Otherwise, we apply the algorithm $\GREEDYP$ in Section~\ref{subsec:greedy}.
The correctness proof is exactly the same, except that the approximation ratio in Lemma~\ref{lm:greedyratio}
changes slightly to $(1-\frac{\varepsilon^{2}}{8})$.

Now we can summarize the overall running time.
Computing the value $\F(\cell,1)$ for each cell requires $O(n\varepsilon^{-6})$ time.
We build $\varepsilon$-approximations in $m$ different cells and $k$ is at most
$\min\{O(1/\varepsilon^{2}),m\}$.
Therefore the overall running time is $O(n(\frac{1}{\varepsilon})^{O(1)}+m(\frac{1}{\varepsilon})^{\Delta})$,
where $\Delta=O(\min(\sqrt{m},\frac{1}{\varepsilon}))$.






\end{document}